\documentclass[aps,pra,10pt,a4paper,reprint,
superscriptaddress,
nofootinbib,twocolumn,longbibliography]{revtex4-1}
\usepackage{amsmath}
\usepackage{amsfonts}
\usepackage{amssymb}
\usepackage{graphicx}
\newcommand{\cH}{\mathcal{H}}
\newcommand{\ot}{\otimes}
\newcommand{\ket}[1]{| #1 \rangle}

\newcommand{\proj}[1]{|#1\rangle\!\langle#1|}
\newcommand{\tr}{\operatorname{tr}}
\usepackage[left=1.5cm,right=1.5cm,top=1.5cm,bottom=1.5cm]{geometry}
\usepackage[dvipsnames]{xcolor}
\usepackage{tikz,amsthm,nicefrac}
\usepackage[colorlinks=true, urlcolor=red, linkcolor=red, citecolor=blue, pdfborder={0 0 0}]{hyperref}

\AtBeginDocument{%
    \newwrite\bibnotes
    \def\bibnotesext{Notes.bib}
    \immediate\openout\bibnotes=\jobname\bibnotesext
    \immediate\write\bibnotes{@CONTROL{REVTEX41Control}}
    \immediate\write\bibnotes{@CONTROL{%
    apsrev41Control,author="08",editor="1",pages="1",title="0",year="0"}}
     \if@filesw
     \immediate\write\@auxout{\string\citation{apsrev41Control}}%
    \fi
  }%
\usetikzlibrary{shapes,positioning,arrows}
\tikzset{>=stealth', c/.style={draw, regular polygon, regular polygon sides=3, minimum height=2.8em, inner sep=0, fill=GreenYellow}, q/.style={draw, circle, minimum height=2.4em, inner sep=0, fill=Melon}, r/.style={draw, rectangle, minimum height=2em, fill= SpringGreen, inner sep=0.5em}, e/.style={<-}, d/.style={-}, f/.style={<->}}
\newtheorem{theorem}{Theorem}

\newtheorem{lemma}{Lemma}
\theoremstyle{definition}
\newtheorem{definition}{Definition}

\usepackage[nodayofweek]{datetime}

\begin{document}
\title{Closing the problem of which causal structures of up to six total nodes have a classical-quantum gap}
\author{Shashaank Khanna}
\email{shashaank.khanna@lis-lab.fr}
\affiliation{Department of Mathematics, University of York, Heslington, York, YO10 5DD, United Kingdom}
\affiliation{Aix-Marseille University, CNRS, LIS, Marseille, France}
\author{Matthew Pusey}
\email{matthew.pusey@york.ac.uk}
\affiliation{Department of Mathematics, University of York, Heslington, York, YO10 5DD, United Kingdom}
\author{Roger Colbeck}
\email{roger.colbeck@kcl.ac.uk}
\affiliation{Department of Mathematics, University of York, Heslington, York, YO10 5DD, United Kingdom}
\affiliation{Department of Mathematics, King's College London, Strand, London, WC2R 2LS, United Kingdom}
\date{\today}
\begin{abstract}
The discovery of Bell that there exist quantum correlations that cannot be reproduced classically is one of the most important in the foundations of quantum mechanics, as well as having practical implications. Bell's result was originally proven in a simple bipartite causal structure, but analogous results have also been shown in further causal structures. Here we study the only causal structure with six or fewer nodes in which the question of whether or not there exist quantum correlations that cannot be achieved classically was open. In this causal structure we show that such quantum correlations exist using a method that involves imposing additional restrictions on the correlations. This hence completes the picture of which causal structures of up to six nodes support non-classical quantum correlations. We also provide further illustrations of our method using other causal structures.
\end{abstract}
\maketitle

\section{Introduction}
Given a set of observed variables, a causal structure encodes a set of dependencies between them and acts as a shorthand for constraints on the joint distribution of those variables. As well as the observed variables a causal structure may also have unobserved systems, whose nature depends on the theory in question. In the classical case, these correspond to random variables, while in the quantum case they may be quantum systems. In this work we are interested in the sets of possible correlations that can be realised in the classical and quantum cases, and, in particular, in whether these sets are different. In cases where they are different we will say that the causal structure supports non-classical quantum correlations, or that it has a classical-quantum gap. 

Bell's theorem~\cite{bell,bellnouvelle} shows the existence of a causal structure supporting non-classical quantum correlations. In fact, using the principle of no fine tuning, which in essence states that one variable can only be the cause of another if the two variables are correlated, Bell's result can be extended to show that there are quantum correlations that cannot be explained classically in any causal structure~\cite{wood}. Fritz~\cite{fritz2012beyond,fritz2016beyond} presented further examples of causal structures with a classical-quantum gap.

We now restrict our attention to the 366565 causal structures with at most six nodes. Henson, Lal and Pusey (HLP)~\cite{HLP_2014} showed that for all but 21 of these there is no classical-quantum gap (or the existence of such a gap can be reduced to one of the 21). 

In~\cite{bell,fritz2012beyond,van2019quantum,chaves2018quantum,lauand2024quantum}, 4 of those 21 causal structures were shown to have a classical-quantum gap, namely the ``Bell"~\cite{bell}, the ``Instrumental"~\cite{van2019quantum,chaves2018quantum}, the ``Triangle"~\cite{fritz2012beyond} and the ``Unrelated Confounders"~\cite{lauand2024quantum} causal structures. Further, in an upcoming work~\cite{upcoming2}, it is shown that 16 of the remaining 17 causal structures also have a classical-quantum gap. The remaining causal structure is shown in Figure~\ref{fig_qc_gap}, and our main result is that it can also support quantum correlations that are non-classical.  We further illustrate our method using the causal structures of Figures~\ref{fig_qc_gap2} and~\ref{fig3}. Therefore, together with the results of~\cite{upcoming2}, we complete the problem of identifying which causal structures of up to six nodes have a classical-quantum gap: the 21 from Appendix~E of~\cite{HLP_2014} (and causal structures that can be reduced to one of them) do, and the remainder do not. 

\section{Classical and quantum correlations within a causal structure}
In this section we give a few definitions to set up our main question.

\begin{definition}[Causal Structure]
A \emph{causal structure} is a directed acyclic graph (DAG) over a set of nodes, each of which is either labelled as latent (denoted here within a circle) or observed (denoted here within a triangle).
\end{definition}
Figure~\ref{fig_qc_gap} gives an example and is the main causal structure we will consider in this work.  Directed edges are interpreted as indicating (the possibility of) direct causal influences. The observed nodes represent things that can be seen and hence each will have an associated random variable\footnote{We use upper case to denote random variables and lower case to denote specific instances of those variables.}. The latent nodes are unobserved and the way they are modelled depends on the underlying theory (in this work we are mainly concerned with classical or quantum theory).

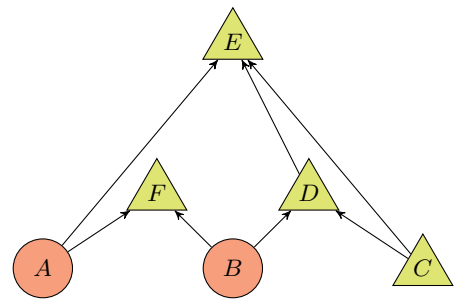
\begin{figure}
    \centering
     \begin{tikzpicture}
            \node[q](B) at (1,0){$B$};
			\node[q](A) at (-1.5,0){$A$};
			\node[c](C) at (3.5,0){$C$};
			\node[c](D) at (2,1){$D$}
			edge[e] (B)
			edge[e] (C);
			\node[c](F) at (0,1){$F$}
			edge[e] (A)
			edge[e] (B);
           \node[c] (E) at (1,3){$E$}
            edge[e] (A)
            edge[e] (C)
            edge[e] (D);
     \end{tikzpicture}
    \caption{The causal structure $G_1$ that is the focus of this work.}
    \label{fig_qc_gap}
\end{figure}

Consider a causal structure $G$ with nodes $X_1,\ldots,X_n$ of which $X_{m+1},\ldots,X_n$ are latent. We use $X_i^{\downarrow}$ to denote the parents of $X_i$ in $G$ (i.e., the nodes from which there is a directed edge to $X_i$ in $G$). In the classical case, each latent node has an associated random variable, and a distribution $P(X_1\ldots X_m)$ over the observed variables is said to be \emph{compatible} with $G$ if there exists a distribution $Q(X_1\ldots X_n)$ over all nodes such that
\begin{eqnarray}
  P(x_1\ldots x_m)&=&\sum_{x_{m+1}\ldots x_n}Q(x_1\ldots x_n),\text{ and}\\
  Q(x_1\ldots x_n)&=&\prod_{i=1}^nQ(x_i|x_i^{\downarrow})\,,\label{eq:prod}
\end{eqnarray}
where we use lower case to denote specific instances of upper case random variables. This condition implies that any variable is conditionally independent of its non-descendants given its parents (Theorem~1.2.7 in~\cite{pearl}), which is known as the causal Markov condition.

Furthermore, this condition implies a number of other conditional independences. These can be conveniently read off the causal structure using the concept of d-separation, which we now introduce.

\begin{definition}[Blocked paths]
Consider a causal structure in which $X$ and $Y$ are nodes and let $Z$ be a set of nodes not containing $X$ or $Y$.  A path from $X$ to $Y$ is said to be \emph{blocked} by $Z$ if it contains either $A\to W\to B$ with $W\in Z$, $A\to W\to B$ with $W\in Z$ or $A\to W\to B$ such that neither $W$ nor any descendant of $W$ belongs to $Z$.
\end{definition}

\begin{definition}[d-separation]
\label{def: d-sep}
Consider a causal structure in which $X$, $Y$ and $Z$ are disjoint sets of nodes.  $X$ and $Y$ are \emph{d-separated} by $Z$, denoted $X\perp Y|Z$ if either there is no path between any node in $X$ and any node in $Y$, or if every path from a node in $X$ to a node in $Y$ is \emph{blocked} by $Z$.
\end{definition}

The following theorem proven in~\cite{Geiger1987,Verma1988} is then useful for inferring additional conditional independences.
\begin{theorem}\label{thm:dsep}
Let $G$ be a classical causal structure and $Q$ be a distribution over all the variables in $G$. $Q$ can be written in the form~\eqref{eq:prod} if and only if for any disjoint sets of nodes $X$, $Y$ and $Z$ in the causal structure with $X\perp Y|Z$, the distribution satisfies $Q(X|YZ)=Q(X|Z)$.
\end{theorem}

In the quantum case, consider a parentless\footnote{We only consider parentless latent nodes here as the more general case is not needed for the present paper.} latent node $X$ with children $X^\uparrow$. For each element $Y$ of $X^\uparrow$ there is a Hilbert space labelled $\cH_X^Y$ and the latent node $X$ is associated with a joint quantum state on the tensor product of these Hilbert spaces, $\bigotimes_{Y\in X^\uparrow}\cH_X^Y$. Each observed node $Y$ with latent parents $Y^L$ and observed parents $Y^O$ has an associated POVM that acts on $\bigotimes_{Q_i\in Y^L}\cH_{Q_i}^Y$, where this POVM can depend on the values of the variables $Y^O$. The joint distribution is formed by applying the POVMs to the quantum state (see Definition~\ref{def:q} for an example). If $Y^L$ is empty, taking the empty tensor product to be $\mathbb{C}$, since a POVM on $\mathbb{C}$ is just a probability distribution, an observed node with no latent parents is assigned a probability distribution (conditioned on the values of any observed parents), recovering the classical case.

In this work we will mainly be interested in the causal structure shown in Figure~\ref{fig_qc_gap}, which we henceforth call $G_1$.  Prior to this work it was not known whether $G_1$ exhibits a gap between the distributions realisable classically and those realisable quantum mechanically, which we now formally define.

\begin{definition}[Classical compatibility with the causal structure $G_1$]\label{def:2}
  A distribution $P(CDEF)$ is said to be classically compatible with $G_1$ if there exists a joint distribution $Q(ABCDEF)$ such that
  \begin{multline}   
    P(cdef)=\sum_{ab}Q(abcdef) \text{\;\;and}\\
    Q(abcdef)= Q(a)Q(b)Q(c)Q(d|bc)Q(e|acd)Q(f|ab)\,.
     \label{eqn1}
	\end{multline}
\end{definition}

\begin{definition}[Quantum compatibility with $G_1$]\label{def:q}
  A distribution $P(CDEF)$ is said to be quantum mechanically compatible with $G_1$ if there exist Hilbert spaces $\cH_A^E$, $\cH_A^F$, $\cH_B^F$, $\cH_B^D$, joint quantum states $\rho_A$ on $\cH_A^E\ot\cH_A^F$ and $\rho_B$ on $\cH_B^F\ot\cH_B^D$, and POVMs $\{K^c_d\}_d$ on $\cH_B^D$ (for each $c$), $\{M^{c,d}_e\}_e$ on $\cH_A^F$ (for each $c,d$) and $\{N_f\}_f$ on $\cH_A^F\ot\cH_B^F$ such that
  \begin{multline}
    P(cdef)=\tr((M^{c,d}_e\ot N_f\ot K^c_d)(\rho_A\ot\rho_B))P(c).
     \label{eqn2}
	\end{multline}
\end{definition}

In the next sections we show that the causal structures of Figures~\ref{fig_qc_gap} and~\ref{fig_qc_gap2} support a classical-quantum gap, and we apply our technique to give an alternative argument for a classical-quantum gap in the Triangle causal structure (Figure~\ref{fig3}).

\section{Results}
Our main result can be stated as follows.
\begin{theorem}
 There exist correlations that are quantum mechanically compatible with $G_1$ but not classically compatible with it.
\label{qctheorem1}
\end{theorem}

The idea behind the proof is to consider correlations satisfying the following additional conditions:
\begin{enumerate}
\item\label{rel1} The variable $F$ breaks down into two, i.e., $F=(F_O,F_S)$, with $F_S$ always equal to $E$.
\item\label{rel2} The variable $E$ is independent of $C$ and $D$, i.e., $P(e|cd)=P(e)$.
\end{enumerate}
In essence, these conditions enforce that $E$ must be determined by $A$ and that the causal influences from $D$ to $E$ and from $C$ to $E$ are not used. In the classical case this then sets up Bell's local causality conditions~\cite{bellnouvelle} with `settings' $E$ and $C$, and `outcomes' $F_O$ and $D$. In the quantum case we can associate an entangled state with node $B$ to violate local causality.

We break down the proof of Theorem~\ref{qctheorem1} into two parts.
\begin{lemma}\label{lem:1}
  Let $P(CDEF)$ be classically compatible with $G_1$ and satisfy Conditions~\ref{rel1} and~\ref{rel2}. Then
  \begin{align}
    P(cdef_O)&=\sum_{ab}Q(abcdef_O)\nonumber\\
             &=\sum_bQ(b)Q(c)Q(e)Q(d|bc)Q(f_O|be).\label{eq:bell}
    \end{align}
\end{lemma}
\begin{lemma}\label{lem:2}
There exist distributions $P(CDEF)$ that are quantum mechanically compatible with $G_1$, satisfy Conditions~\ref{rel1} and~\ref{rel2} but do not satisfy Equation~\eqref{eq:bell}.
\end{lemma}

\begin{proof}[Proof of Lemma~\ref{lem:1}]   
Since $P(CDEF)$ is classically compatible with $G_1$ we can write
  \begin{multline*}
    P(cdef)=\sum_{ab}Q(abcdef)
\\=\sum_{ab}Q(b)Q(c|b)Q(d|bc)Q(e|bcd)Q(f|bcde)Q(a|bcdef),
  \end{multline*}
where we have used the definition of conditional probability to expand $Q$ in the second line.  The following d-separation relations hold for $G_1$: $F\perp CD|B$, $E\perp B|CD$ and $B\perp C$. Theorem~\ref{thm:dsep} hence implies $Q(f|bcd)=Q(f|b)$, $Q(e|bcd)=Q(e|cd)$, and $Q(c|b)=Q(c)$. Using the last two of these we have
  \begin{align}
    P(cdef)=\sum_bQ(b)Q(c)Q(d|bc)Q(e|cd)Q(f|bcde).
  \end{align}
Condition~\ref{rel2} allows $Q(e|cd)$ to be replaced with $Q(e)$, and we then use Condition~\ref{rel1} to replace $f$ with $f_Sf_O$ and sum both sides over $f_S$ to give
  \begin{align}\label{eq:cdefO}
    P(cdef_O)=\sum_bQ(b)Q(c)Q(d|bc)Q(e)Q(f_O|bcde).
  \end{align}
Furthermore, using that $E$ equals $F_S$ (Condition~\ref{rel1}), we write $Q(ef_O|bcd)=Q(f_O f_S|bcd)=Q(f|bcd)=Q(f|b)=Q(ef_O|b)=Q(e|b)Q(f_O|be)=Q(e)Q(f_O|be).$\footnote{Where $Q(e|b)=Q(e)$ follows from $Q(e|bcd)=Q(e|cd)=Q(e)$.} Also noting that $Q(ef_O|bcd)=Q(f_O|bcde)Q(e)$, we get $Q(f_O|bcde)=Q(f_O|be)$ and substituting this into~\eqref{eq:cdefO} we recover~\eqref{eq:bell}.
\end{proof}
\begin{proof}[Proof of Lemma~\ref{lem:2}]
  We consider the following model, where, for brevity, we define $\ket{\theta}=\cos(\theta)\ket{0}+\sin(\theta)\ket{1}$:
  \begin{enumerate}
  \item Take $A$ to be a uniformly distributed classical bit\footnote{This could be put in quantum form by considering a state $\frac{1}{2}(\proj{00}+\proj{11})$ and using a measurement in the $\{\ket{0},\ket{1}\}$ basis whenever we want to reveal its value.}.
  \item Take $\cH^{F_O}_B$ and $\cH^D_B$ to be two-dimensional Hilbert spaces and $\rho_{\cH^{F_O}_B\cH^D_B}$ to be the projector onto $\frac{1}{\sqrt{2}}\left(\ket{00}+\ket{11}\right)$.
    \item Take $C$ to be a uniformly distributed classical bit.
  \item \label{it:q1} Take $F_S=A$. When $A=0$, take $F_O$ to be the outcome of a measurement in the $\{\ket{\theta=0},\ket{\theta=\pi/2}\}$ basis on $\cH^{F_O}_B$. When $A=1$, take $F_O$ to be the outcome of a measurement in the $\{\ket{\pi/4},\ket{3\pi/4}\}$ basis on $\cH^{F_O}_B$.
  \item \label{it:q2} When $C=0$, take $D$ to be the outcome of a measurement in the $\{\ket{\pi/8},\ket{5\pi/8}\}$ basis on $\cH^D_B$. When $C=1$, take $D$ to be the outcome of a measurement in the $\{\ket{-\pi/8},\ket{3\pi/8}\}$ basis on $\cH^D_B$.
    \item Take $E=A$.
  \end{enumerate}
  
\begin{table}[t]
\begin{tabular}{cccc|c}
\hphantom{0}$c$\hphantom{0} & \hphantom{0}$e$\hphantom{0} & \hphantom{.}$f_O$\hphantom{.} & \hphantom{0}$d$\hphantom{0} & \hphantom{0}$P(c\,e\,f_O\,d)$\hphantom{0} \\ \hline
0 & 0 & 0 & 0 & $\frac{1}{8}\cos^2(\frac{\pi}{8})$\\
0 & 0 & 0 & 1 & $\frac{1}{8}\sin^2(\frac{\pi}{8})$\\
0 & 0 & 1 & 0 & $\frac{1}{8}\sin^2(\frac{\pi}{8})$\\
0 & 0 & 1 & 1 & $\frac{1}{8}\cos^2(\frac{\pi}{8})$\\
0 & 1 & 0 & 0 & $\frac{1}{8}\cos^2(\frac{\pi}{8})$\\
0 & 1 & 0 & 1 & $\frac{1}{8}\sin^2(\frac{\pi}{8})$\\
0 & 1 & 1 & 0 & $\frac{1}{8}\sin^2(\frac{\pi}{8})$\\
0 & 1 & 1 & 1 & $\frac{1}{8}\cos^2(\frac{\pi}{8})$\\
1 & 0 & 0 & 0 & $\frac{1}{8}\cos^2(\frac{\pi}{8})$\\
1 & 0 & 0 & 1 & $\frac{1}{8}\sin^2(\frac{\pi}{8})$\\
1 & 0 & 1 & 0 & $\frac{1}{8}\sin^2(\frac{\pi}{8})$\\
1 & 0 & 1 & 1 & $\frac{1}{8}\cos^2(\frac{\pi}{8})$\\
1 & 1 & 0 & 0 & $\frac{1}{8}\sin^2(\frac{\pi}{8})$\\
1 & 1 & 0 & 1 & $\frac{1}{8}\cos^2(\frac{\pi}{8})$\\
1 & 1 & 1 & 0 & $\frac{1}{8}\cos^2(\frac{\pi}{8})$\\
1 & 1 & 1 & 1 & $\frac{1}{8}\sin^2(\frac{\pi}{8})$\\
\end{tabular}
\caption{The distribution generated by the quantum correlations. 
We omit a column for $f_S$ because $f_S=e$.}
\label{table1}
\end{table}

By construction these correlations satisfy relations~\ref{rel1} and~\ref{rel2}.

Bell's theorem states that if~\eqref{eq:bell} holds, then $P(DF_O|CE)$ must satisfy the CHSH inequality~\cite{PhysRevLett.23.880}, which can be expressed as
$P(D=F_O|00)+P(D=F_O|01)+P(D=F_O|10)+P(D\neq F_O|11)\leq3$.\footnote{Here $P(D=F_O|ce)=\sum_xP(D=x,F_O=x|ce)$.}
However, the correlations given in Table~\ref{table1} give $P(D=F_O|00)+P(D=F_O|01)+P(D=F_O|10)+P(D\neq F_O|11)=4\cos^2(\frac{\pi}{8})\approx3.41$, in violation of the CHSH inequality, and hence cannot be realised classically in $G_1$.
\end{proof} 

\section{Other causal structures}
The approach presented here can be applied to other causal structures. One case that follows directly from the above analysis is the causal structure $G_2$ in Figure~\ref{fig_qc_gap2}, which is the same as $G_1$ but without an arrow from $C$ to $E$. Since the quantum correlations we constructed in Lemma~\ref{lem:2} for the DAG $G_1$ do not use the causal link from $C$ to $E$, the same correlations are quantum compatible with $G_2$. Any distribution that is classically compatible with $G_2$ is also classically compatible with $G_1$ (removing an arrow in a causal structure adds additional constraints on the correlations). Hence, the quantum correlations are not classically compatible with $G_2$. The existence of a classical-quantum gap for $G_2$ is also shown in~\cite{upcoming2} using another method.

 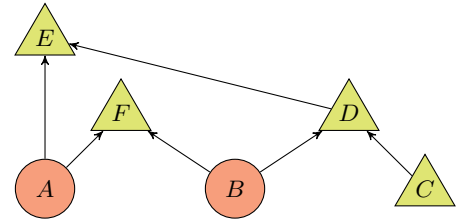
\begin{figure}
    \centering
     \begin{tikzpicture}
            \node[q](B) at (1,0){$B$};
			\node[q](A) at (-1.5,0){$A$};
			\node[c](C) at (3.5,0){$C$};
			\node[c](D) at (2.5,1){$D$}
			edge[e] (B)
			edge[e] (C);
			\node[c](F) at (-0.5,1){$F$}
			edge[e] (A)
			edge[e] (B);
           \node[c] (E) at (-1.5,2){$E$}
            edge[e] (A)
            edge[e] (D);
     \end{tikzpicture}
    \caption{The causal structure $G_2$ which also has a classical-quantum gap by our argument.}
    \label{fig_qc_gap2}
\end{figure}

The argument we use in this work can be thought of as an extension of that used by Fritz~\cite{fritz2012beyond} to establish a classical-quantum gap in the triangle causal structure shown in Figure~\ref{fig3}. A distribution is classically compatible with the triangle causal structure if there exists a distribution $Q(ABCDEF)$ such that
\begin{align*}                                  P(def)&=\sum_{abc}Q(abcdef),\text{ where}\\Q(abcdef)&=Q(a)Q(b)Q(c)Q(d|bc)Q(e|ac)Q(f|ab).
\end{align*}
Fritz then takes the additional conditions that $D$, $E$ and $F$ are two bits each, i.e., $D=(S_1,O_1)$, $E=(S'_1,S'_2)$, $F=(S_2,O_2)$, and that $S'_1=S_1$ and $S'_2=S_2$, the idea being that these imply that $S_1$ must have its causal origin in $C$ and $S_2$ must have its causal origin in $A$.
\begin{lemma}\label{lem:3}
  Let $P(DEF)$ be classically compatible with the triangle causal structure and satisfy the additional conditions above. It follows that
  \begin{multline}\label{eq:belltr}
P(s_1s_2o_1o_2)=\sum_{abc}Q(s_1s_2o_1o_2abc)\\=\sum_bQ(b)Q(s_1)Q(s_2)Q(o_1|s_1b)Q(o_2|s_2b).
\end{multline}
\end{lemma}
Note that~\eqref{eq:belltr} is Bell's local causality condition for settings $S_1$ and $S_2$ and outcomes $O_1$ and $O_2$. The proof we give is different to that in~\cite{fritz2012beyond}.
\begin{proof}
Since $P(S_1S_2O_1O_2)$ is classically compatible with the triangle causal structure we can write
\begin{multline*}
P(s_1s_2o_1o_2)=\sum_{abc}Q(s_1s_2o_1o_2abc)\\=
\sum_bQ(b)Q(s_1|b)Q(s_2|bs_1)Q(o_1|s_1s_2b)Q(o_2|s_1s_2o_1b),
\end{multline*}
where we have used the definition of conditional probability to expand $Q$ in the second line. Since $S_1=S_1'$, we have $Q(s_1|b)=Q(s_1'|b)$, but $B\perp E$ and Theorem~\ref{thm:dsep} means $Q(s_1'|b)=Q(s_1')$, and hence $Q(s_1|b)=Q(s_1)$. Likewise, $Q(s_2|bs_1)=Q(s_2|b)=Q(s_2'|b)=Q(s_2')=Q(s_2)$ follows from $F\perp D|B$, $S_2=S_2'$ and $B\perp E$. Similarly, $D\perp F|B$ gives both $Q(s_1o_1|s_2o_2b)=Q(s_2o_2|b)$ and $Q(s_2o_2|s_1o_1b)=Q(s_2o_2|b)$, which imply $Q(o_1|s_1s_2b)=Q(o_1|s_1b)$ and $Q(o_2|s_1s_2o_1b)=Q(o_2|s_2b)$. We hence reduce to~\eqref{eq:belltr}.
\end{proof}
Thus, as in the causal structure $G_1$ considered earlier, any quantum distribution that violates the CHSH inequality in the Bell causal structure can be used to show a classical-quantum gap in the triangle causal structure, by taking $A$ and $C$ binary uniform bits, $B$ to be an entangled state, $S_1=S_1'=C$, $S_2=S_2'=A$ and forming $O_1$ and $O_2$ using measurements on the entangled state with the choice of measurement depending on $C$ and $A$ respectively (in the same way as in points~\ref{it:q1} and~\ref{it:q2} of the quantum strategy for $G_1$).

 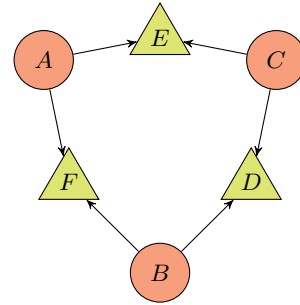
\begin{figure}
    \centering
     \begin{tikzpicture}
           [scale=0.6]
				\node[q](L) at (-0.55,2.7){$A$};
				\node[q](Z) at (2,-2){$B$};
				\node[q](O) at (4.55,2.7){$C$};
				\node[c](A) at (0,0){$F$}
				edge[e] (L)
				edge[e] (Z);
				\node[c](C) at (4,0){$D$}
				edge[e] (O)
				edge[e] (Z);
				\node[c](D) at (2,3.2){$E$}
				edge[e] (O)
				edge[e] (L);
     \end{tikzpicture}
    \caption{The triangle causal structure. For the argument we will take $D=(S_1,O_1)$, $E=(S_1',S_2')$ and $F=(S_2,O_2)$.}
    \label{fig3}
\end{figure}
Note that the quantum distribution used to show a gap in the triangle causal structure has $A$ and $C$ classical, and hence there is also a classical-quantum gap in variations of the triangle where either $A$ or $C$ are observed. 

\section{Conclusion}
We have shown that there are quantum correlations that cannot be classically realised in the causal structure $G_1$, the only remaining open case of up to six nodes. Our method corresponds to an extension of an idea used by Fritz in~\cite{fritz2012beyond} for answering the same question in the triangle causal structure, and we include a short argument for the triangle causal structure as a further illustration. Our arguments work at the level of the probabilities, in contrast to Fritz's use of entropies. 

Given a causal structure, we can also consider post-quantum correlations (those that obey the independence relations among the observed nodes that hold in the classical case but are not quantum compatible). The proof of HLP that all but 21 of the causal structures with at most six nodes have no classical-quantum gap more generally shows that all apart from those 21 have no post-quantum correlations, leaving 21 candidates (up to reductions) that could support post-quantum correlations. The results of the present paper, along with~\cite{bell,fritz2012beyond,HLP_2014,van2019quantum,chaves2018quantum,lauand2024quantum,upcoming2}, imply that, for up to six nodes, there are no causal structures with post-quantum correlations but without non-classical quantum correlations.

\bigskip

\noindent{\bf Additional note:} A preliminary version of this work is available in SK's PhD Thesis~\cite{shashaankthesis}.

\begin{acknowledgments}
    We thank Elie Wolfe for making us aware of the results of~\cite{upcoming2}. SK was supported by a studentship from the Department of Mathematics, University of York.
\end{acknowledgments}

\end{document}